\documentclass[a4paper,accepted=2026-03-14]{quantumarticle}
\pdfoutput=1

\usepackage[colorlinks]{hyperref}
\usepackage[numbers,sort&compress]{natbib}
\usepackage[phfqit]{tchshortcuts}
\usepackage{amsthm}
\usepackage[
	print-unity-mantissa=false,
	exponent-product={\cdot},
	range-phrase=--
]{siunitx}
\usepackage[compat=0.6]{yquant}
\usepackage{relsize}
\usepackage{booktabs}
\usepackage{algorithm}
\usepackage{algpseudocodex}
\usepackage[capitalise,nameinlink]{cleveref}
\usepackage{tikz}
\usepackage{tabularx}

\DeclareMathOperator{\wh}{wh}

\newtheorem{cor}{Corollary}
\newtheorem{prop}{Proposition}
\newtheorem{lem}{Lemma}
\theoremstyle{definition}
\newtheorem{defn}{Definition}
\newtheorem{cond}{Condition}
\theoremstyle{remark}
\newtheorem{rem}{Remark}

\crefname{step}{step}{steps}

\yquantset{/yquant/operator/minimum width=0mm}

\usetikzlibrary{positioning, arrows.meta}

\newcommand\oxfordaddress{Department of Materials,
University of Oxford,
Parks Road,
Oxford OX1 3PH,
United Kingdom}

\newcommand\qiqbaddress{Center for Quantum Information and Quantum Biology,
The University of Osaka,
1-2 Machikaneyama,
Toyonaka 560-0043,
Japan}

\begin{document}

\title{Snowflake: A Distributed Streaming Decoder}

\author{Tim Chan}
\email{timothy.chan@materials.ox.ac.uk}
\affiliation{\oxfordaddress}
\affiliation{\qiqbaddress}
\orcid{0000-0001-6187-7402}

\begin{abstract}
	We design \emph{Snowflake},
	a quantum error correction decoder
	that,
	for the surface code under circuit-level noise,
	is roughly 25\% more accurate than
	the Union--Find decoder,
	with a better mean runtime scaling:
	subquadratic as opposed to cubic in the code distance.
	Our decoder runs in a streaming fashion
	and has a distributed, local implementation.
	In designing Snowflake,
	we propose a new method for general stream decoding
	that eliminates the processing overhead due to
	window overlap in existing windowing methods.
\end{abstract}

\maketitle
\section{Introduction}
\label{sec:introduction}
Decoders are crucial to the quantum error correction
that enables fault-tolerant quantum computing.
These classical algorithms must decode
(i.e.\ infer errors from measurements and fix them)
accurately and extremely quickly to prevent the quantum computation from corrupting.
It is difficult to design a decoder sufficiently
	accurate,
	fast,
	and practical to implement in hardware \cite{Reilly2019,Das2022,Delfosse2023a}.
This last requirement has led to interest in decoders that are local
i.e.\ running on a grid of identical processors,
each communicating only with their nearest neighbours
\cite{Harrington2004,
Fowler2015,
Herold2015,
Breuckmann2017,
Herold2017,
Lang2018,
Kubica2019,
Delfosse2020a,
Smith2023,
Lake2025,
Lake2025a,
McArdle2025}.
Such decoders benefit from
	parallelism \cite[\S III.D]{Terhal2015},
	repairability \cite[\S 3]{Chan2023c},
	reduced signal losses \cite[pp~\numrange{6}{7}]{Vandersypen2017},
	and lower latency \cite[\S VI]{Reilly2019} \cite[\S 4.2]{Battistel2023}
when used with local
(requiring no long-range qubit--qubit interactions)
codes.

A relatively new candidate gaining in popularity
\cite{Li2019,
Huang2020,
Hu2020a,
Delfosse2021a,
Pattison2021,
Delfosse2022,
Das2022,
Berent2023,
Barber2025,
Griffiths2024,
Lobl2024,
Meister2024}
is the Union--Find decoder (UF)
\cite{Delfosse2020,Delfosse2021}.
Recent literature has explored UF implementations of varying levels of locality
\cite{Liyanage2024,Heer2023,Heer2023a,Chan2023c,Ziad2025,Valentino2024,Gu2025,Kishi2026},
the first being \emph{Helios} \cite{Liyanage2023a}.
The closely related \emph{Macar} \cite{Chan2023c}
exhibits very similar runtime behaviour to Helios,
and is the implementation we will repeatedly refer to for comparison.
These implementations,
save Helios,
tacked only the batch decoding problem:
a simplification of the stream decoding problem
that actually occurs in practice.
Relatively little work has gone into the latter:
since 2002,
two methods have been proposed allowing most batch decoders
to be used as stream decoders.
These are
	the commonly used \emph{forward method} \cite[\S VI.B]{Dennis2002}
	and the newer \emph{sandwich method} \cite{Skoric2023,Tan2023},
using the terminology of \cite{Tan2023}.

\begin{figure}
	\centering
	\includegraphics[width=0.48\textwidth]{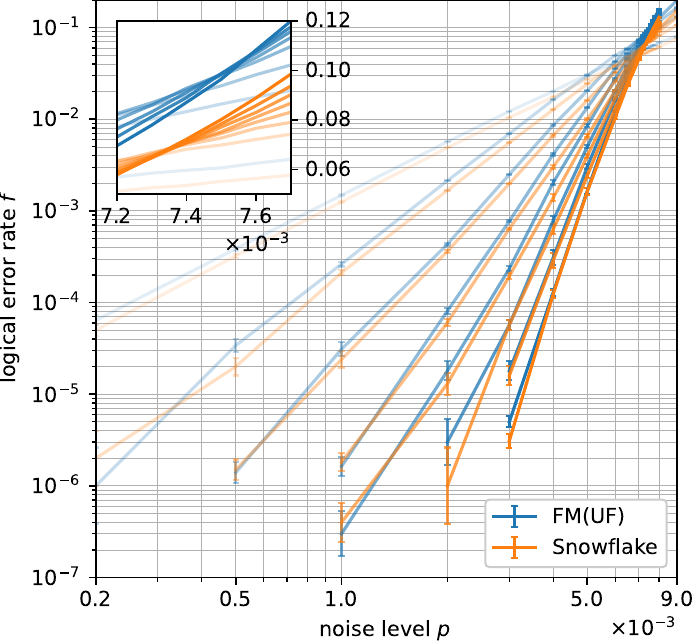}
	\caption{Threshold plots for the surface code
	under circuit-level noise.
	Each line corresponds to a fixed distance
	$d =3, 5, \dots, 17$ from shallowest to steepest.
	Each datapoint is the mean of
	\numrange{1e4}{1e7} lots of $d$ measurement rounds;
	errorbars show standard error.
	The inset zooms in on both thresholds
	whose coordinate $(p, f)_\th$ is around
	$(0.755, 10.5)\cdot\num{1e-2}$ for UF adapted with the forward method, and
	$(0.735, 6.8)\cdot\num{1e-2}$ for Snowflake.
	Here,
	the standard error is smaller than the line thickness.}
	\label{fig:FM_UF_v_Snowflake}
\end{figure}

We design Snowflake,
a generalisation of the local decoders in \cite{Chan2023c}
from the batch to the streaming case.
Snowflake comprises a simple set of local rules
with very little central management,
and works for any
error-correcting code with a decoding graph
embeddable in $\mathbb R^3$
without `long-range' edges.
Examples include the repetition code
and the surface code under circuit-level noise.
The pseudocode for Snowflake is summarised in \cref{alg:2_1_schedule}.

Snowflake relies on the new
\emph{frugal method} we develop for general stream decoding.
This is similar to the forward method,
but differs in that it discards none of the computation --
it is more frugal with resources.
This halves both
	the power draw
	and the size
of the decoder architecture.
We hope the proposal of this method
enables the development of more streaming decoders.

\Cref{fig:FM_UF_v_Snowflake} shows that,
when tested on the surface code under circuit-level noise,
Snowflake recovers \num{\approx 97}\% of the threshold of
UF adapted with the forward method
[which we abbreviate as FM(UF)],
but in absolute terms is \num{24.9(5)}\% more accurate.
Moreover,
Snowflake's mean runtime scaling with the code distance
improves upon original UF.
Snowflake has a distributed implementation
that is modular and scalable.
This is a 3D lattice of processors
though we note an alternative architecture
that `flattens' this lattice onto a 2D grid
would be even more practical given current chip design capabilities,
at a potentially marginal cost in speed.

\Cref{sec:background} covers the prerequisite theory.
We describe the frugal method in \cref{sec:the_frugal_method},
explain Snowflake in \cref{sec:snowflake},
then numerically test its accuracy and runtime in \cref{sec:numerics}.
\Cref{sec:conclusion} concludes.
In \crefrange{fig:add_layer}{fig:distributed_implementation}
and \crefrange{fig:drop_grow_merging}{fig:arbitrary_choice},
the time axis points upward.
Our emulation code is on GitHub at \cite{Chan2023a_quantum_bibstyle}.

\section{Background}
\label{sec:background}
\Cref{sec:the_detector_error_model,sec:the_decoding_graph} introduce
widely used formalisms for modelling noisy circuits.
These two subsections can be skimmed
for readers familiar with error correction.
\Cref{sec:the_batch_decoding_problem}
recaps the batch decoding problem.
In \cref{sec:the_stream_decoding_problem}
we generalise to the streaming case
and in \cref{sec:the_forward_method}
we provide one solution that already exists in literature.
Throughout the paper,
we denote the node and edge sets of any graph $\Gamma$
as $V_\Gamma$ and $E_\Gamma$, respectively.

\subsection{The Detector Error Model}
\label{sec:the_detector_error_model}

We can model noise affecting the quantum error-correction code
using a \emph{detector error model} (DEM) \cite{Higgott2025,Piveteau2024}.
Constructing a DEM starts with defining a set of \emph{detectors} and \emph{logical observables},
each of which represents a set of measurements whose parity is deterministic in the absence of noise;
e.g.\ a detector could represent two consecutive measurements of the same stabiliser
at a given point in time.
Noise is then modelled by a set of independent events called \emph{error mechanisms},
each of which flips a unique subset of detectors and logical observables.
For stabiliser measurement
under the phenomenological noise model,
	error mechanisms either a data qubit bitflip or a measurement outcome flip;
under a more realistic circuit-level noise model,
	they represent sets of error events within the physical circuit.
\Cref{sec:noise_model} describes the circuit-level noise model used in our numerics.

A DEM of $m$ detectors, $l$ logical observables, and $n$ error mechanisms
is specified by a triple $(H, L, \v p)$,
where the matrix $H \in \mathbb F_2^{m \times n}$ $(L \in \mathbb F_2^{l \times n})$
defines which error mechanism flips which detector (logical observable),
and the vector $\v p \in {]0, 1[}^n$
defines the probability each error mechanism occurs;
this is usually a function of some universal noise level $p$.

In a decoding cycle,
the error mechanisms occur with probability according to $\v p$;
the occurred set is represented by a vector $\v e \in \mathbb F_2^n$.
This flips a set of detectors and logical observables
indicated by the support of $\v s =H\v e$ and $\v \lambda =L\v e$, respectively.
The task of the decoder is to infer $\v \lambda$ given $\v s$.
More generally,
we will use the following definition.
\begin{defn}\label{defn:decoding_cycle}
A \emph{decoding cycle} is one iteration of the process
from when the decoder receives a new (or updated) $\v s$
to when it outputs a new (or updated) $\v \lambda$.
\end{defn}

In this paper
we focus on the $l =1$ case;
specifically,
the Z memory experiment
i.e.\ inferring whether the $\hat \Z$ observable of
one logical qubit has flipped
after some number of stabiliser measurement rounds.
We consider decoders that output a specific correction vector
$\v c \in \mathbb F_2^n$
satisfying
\begin{equation}\label{eq:correction_vector_requirement}
H\v c =\v s;
\end{equation}
the $\v \lambda$ inferred is then $\widehat{\v \lambda} =L\v c \in \mathbb F_2^1$.

\subsection{The Decoding Graph}
\label{sec:the_decoding_graph}
Additionally,
we focus on DEMs in which each error mechanism flips at most two detectors
(or DEMs that can be approximated as such).
In this special case,
we can construct a \emph{decoding graph} $G$ from the DEM.
We start by representing each detector as a node.
Then for each error mechanism:
if it flips exactly two detectors $`{u, v}$,
we add an edge $uv$;
if it flips exactly one detector $u$,
we add an edge between $u$ and a dedicated \emph{boundary node}
which represents nothing physical,
but exists solely to allow the error mechanism to be represented by an edge.

Thus,
$V_G$ is partitioned into two sets:
\emph{detectors} $V_\d$ and \emph{boundary nodes} $V_\b$.
Each node is assigned a unique spacetime coordinate
such that for codes with local stabilisers (e.g.\ the surface code),
there are no `long-range' edges.
This locality is useful for
local decoders that process information between nearby detectors,
like Snowflake.
$V_\b$ is further partitioned
into the \emph{west boundary} and the \emph{east boundary} sets,
chosen such that an edge flips the logical observable
iff it is incident to the west boundary.

Since $e \in E_G$
represents an error mechanism which can occur or not,
we assign $e$ a bit value $x \in `{0, 1}$
and may refer to it as a \emph{bit value-$x$ edge}.
The bit value $x$ is not known to the decoder,
but it \emph{can} flip $x$:
we call this \emph{flipping the edge}.
Of course,
the decoder can
keep track of which edges it has flipped.
A \emph{defect} is a detector incident to
an odd number of bit value-1 edges.
The \emph{syndrome} $\mathbb S \subseteq V_\d$ is the set of defects.
The decoder can modify $\mathbb S$:
flipping an edge $uv$ updates
$\mathbb S \gets \mathbb S \sd `{u, v}$
where $\sd$ denotes symmetric difference.
\begin{defn}\label{defn:push_and_annihilate}
In the case $u \in \mathbb S$ and $v \notin \mathbb S$,
we say flipping $uv$
will \emph{push} the defect at $u$ along $uv$ to $v$.
In the case $u, v \in \mathbb S$,
we say flipping $uv$
will \emph{annihilate} the defect at $u$ with the defect at $v$.
We extend this to the case of a path $P$ of edges
connecting a defect at $s$ to another defect at $t$ (or to a boundary node):
flipping the edges in $P$ will push the defect at $s$ along $P$
to annihilate with the defect at $t$ (or with the boundary node).
\end{defn}

\begin{defn}\label{defn:annihilate_all_defects}
A set $F \subseteq E_G$ \emph{annihilates} syndrome $\mathbb S$
if $F$ constitutes paths that pair each defect with either
another defect or a boundary node.
\end{defn}

\subsection{The Batch Decoding Problem}
\label{sec:the_batch_decoding_problem}
In terms of the decoding graph,
the decoding cycle is as follows.
Each edge is assigned bit value 1 with probability according to $\v p$.
Knowing only $\mathbb S$,
the decoder must flip edges
until $\mathbb S =\varnothing$.
\begin{defn}\label{defn:correction}
The \emph{correction} $\mathbb C \subseteq E_G$
is the final set of edges the decoder chooses to flip.
\end{defn}
\noindent
Flipping edges until $\mathbb S =\varnothing$ means
$\mathbb C$ must annihilate $\mathbb S$;
this is equivalent to \cref{eq:correction_vector_requirement},
where $\mathbb S$ and $\mathbb C$
are analogous to $\v s$ and $\v c$.
After applying $\mathbb C$,
the remaining bit value-1 edges can only constitute either cycles,
or paths that end at boundaries.
Such a path will flip the logical observable
if it ends at the west boundary exactly once;
this motivates the following.

\begin{defn}\label{defn:logical_error_rate}
A \emph{logical flip} is path of bit value-1 edges
from the west to the east boundary.
The \emph{logical error rate}
is the logical flip count
per $d$ measurement rounds~\cite[\S VI]{Stephens2014},
where $d$ is the code distance.
\end{defn}
\noindent
The shortest logical flip has $d$ edges.
We evaluate decoder accuracy
by calculating the logical error rate.

\subsection{The Stream Decoding Problem}
\label{sec:the_stream_decoding_problem}

\begin{figure}
	\centering
	\includegraphics[width=0.48\textwidth]{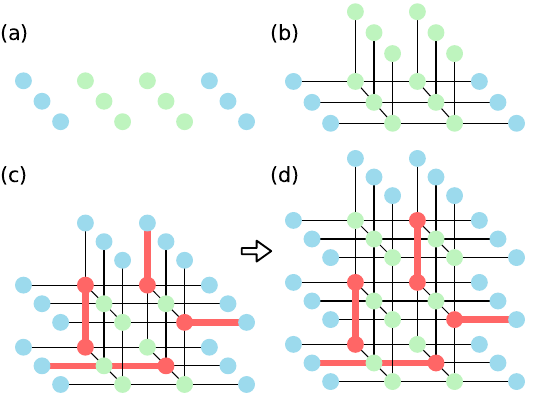}
	\caption{Defects are red nodes;
	other detectors, green;
	boundary nodes, blue.
	Bit value-1 edges are thick and red;
	bit value-0 edges, thin and black.
	(a) A sheet and (b) a layer
	of (c) the decoding graph for the surface code
	under phenomenological noise.
	The west, east, and future boundaries
	are the set of blue nodes on the
	left, right, and upper faces of the decoding graph,
	respectively.
	(d) How said graph
	changes after one measurement round:
	one layer is added on top.
	This is the streaming generalisation of \cite[Figure~3b]{Chan2023c}.
	}
	\label{fig:add_layer}
\end{figure}

We now generalise the batch decoding problem
to the streaming case,
relevant for very long memory experiments.
\Cref{fig:add_layer}a,b illustrates the following definitions:
	a \emph{sheet} is a subset of $V_G$ all with the same $t$ coordinate,
	and a \emph{layer} is the periodic unit subgraph of $G$.
For the stream decoding problem,
there are an \emph{indefinite} number of measurement rounds;
this means new layers are constantly being added
to the top of $G$ as in \cref{fig:add_layer}c,d.
Note the addition of a \emph{future boundary}:
a third set of boundary nodes representing detectors
whose parities are yet to be determined.
When a new layer is added,
its detectors replace the future boundary of the previous layer.
In general,
some edges in the new layer are assigned bit value 1
(specifically,
the probability that $x =1$ for each edge
equals that of the analogous edge in the layer below),
creating defects in the new layer.
The decoder must constantly add to $\mathbb C$
to annihilate these new defects.

At the end of the Z memory experiment,
all data qubits are destructively measured in the Z basis.
This corresponds the final layer that is added onto $G$,
and it has no future boundary.
Since the final form of $G$ has only a west and east boundary,
we can once again
evaluate decoder accuracy
by \cref{defn:logical_error_rate}.

For clarity,
all figures in the rest of this paper,
save \crefrange{fig:distributed_implementation}{fig:stage_flowchart},
use a planar $G$ corresponding to the repetition code
under phenomenological noise.
However,
the same concepts apply to other graphs
such as for the surface code under circuit-level noise.
In the next subsection
we discuss an existing approach to solve the stream decoding problem.

\subsection{The Forward Method}
\label{sec:the_forward_method}
This method was first proposed in
	\cite[\S VI.B]{Dennis2002} as the `overlapping recovery method'
but a more detailed explanation can be found in
	\cite[pp~\numrange{2}{3}]{Skoric2023} where it is called `sliding window decoding'.
It has been applied to the general class of \emph{quantum low-density parity-check codes}
(LDPC, which includes the repetition and surface codes)
\cite{Berent2024,Huang2024,Gong2024,Kuo2024}.
It reduces the stream decoding problem to batch decoding
by chunking the syndrome,
decoding each chunk,
then stitching together all the corrections.
The chunks must overlap,
else a small defect pair could straddle two consecutive chunks
hence be decoded inaccurately.
\begin{figure}
	\centering
	\includegraphics[width=0.48\textwidth]{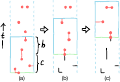}
	\caption{Two decoding cycles of the forward method.
	The \emph{decoding window} (\cref{defn:decoding_window})
	is the solidly outlined rectangle
	with boundaries in blue.
	Defects are red dots.
	The \emph{tentative correction} (\cref{defn:tentative_correction}) comprises the red lines;
	edges flipped by the decoder, black lines.
	The commit and buffer region heights
	are $c$ and $b$, respectively.
	(a) The tentative correction contains a path that spans both regions.
	(b) This path is partially committed
	which makes an `artificial defect' at the bottom of the new window
	that must also be annihilated.
	(c) The rest of the path is committed.}
	\label{fig:forward_method}
\end{figure}

\begin{figure*}
	\centering
	\includegraphics[width=\textwidth]{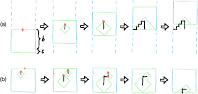}
	\caption{Two examples of how Snowflake finds its tentative correction.
	Features are shown as in \cref{fig:forward_method} and additionally,
	\emph{clusters} (\cref{defn:cluster}) are outlined in green.
	(a) A cluster grows around one defect as the decoding window is raised,
	while the defect is always \emph{pushed}
	(\cref{defn:push_and_annihilate})
	to the highest point.
	As soon as the cluster touches a boundary,
	the defect is instead pushed to that boundary
	(in this case,
	the cluster touches the west and east boundaries at the same time
	and arbitrarily chooses to push westward).
	This \emph{inactivates} (\cref{defn:inactive_cluster}) the cluster so it stops growing.
	It gradually disappears from view.
	The defect's path is staircase-like
	as it comprises edges from the decoding graph.
	(b) Two defects separated in time and space
	(this time viewed from the perspective of the window
	so clusters appear to fall).
	The lower cluster is older,
	hence larger,
	when both clusters merge.
	Frame 3 shows the lower defect on its way to the highest point
	to annihilate with the other defect,
	inactivating the resultant cluster.}
	\label{fig:snowflake}
\end{figure*}

More concretely,
assume $G$ already comprises many layers,
and define the following.
\begin{defn}\label{defn:decoding_window}
The \emph{decoding window} is at any given time
a graph $W =(C \cup B) \subseteq G$
comprising $c +b$ consecutive layers of $G$,
wherein the highest sheet of detectors is replaced by a future boundary.
The lowest $c$ layers of $W$ is the
\emph{commit region} $C$;
the highest $b$ layers of $W$ is the
\emph{buffer region} $B$.
\end{defn}
\begin{defn}\label{defn:tentative_correction}
The \emph{tentative correction} $\mathbb T \subseteq E_W$
is the decoder's output for a given decoding window $W$.
This differs from \cref{defn:correction} in that $\mathbb T$ may not be final.
We say we \emph{commit} edges from $\mathbb T$
when we transfer them to $\mathbb C$.
\end{defn}
\noindent
The forward method is
illustrated in \cref{fig:forward_method}.
$W$ starts at the bottom,
i.e.\ as the lowest $c +b$ layers,
of $G$,
and $\mathbb C =\varnothing$ initially.
The method then indefinitely repeats the
following (skipping \cref{step:raise} in the first iteration).
\begin{algorithm}[H]
\caption{The forward method decoding cycle.}
\label{alg:forward_method_decoding_cycle}
\begin{algorithmic}[1]
	\State Raise $W$ by $c$ layers
	i.e.\ redefine it as the subgraph of $G$
	that is $c$ layers higher.
	\label[step]{step:raise}
	\State Use any batch decoder to find a tentative correction $\mathbb T$
	that annihilates $\mathbb S \cap V_W$.
	\State Update
	$\mathbb C \gets \mathbb C \cup (\mathbb T \cap E_C)$
	and forget $\mathbb T \cap E_B$.
\end{algorithmic}
\end{algorithm}
\begin{rem}\label{rem:forward_method_discards}
Since $b$ `layers of computation' are forgotten per decoding cycle,
decoding $\tau$ sheets of syndrome data discards
a total of $\mathord{\approx} b \tau/c$ layers of computation.
\end{rem}
\begin{rem}\label{rem:small_buffer_or_commit_height}
Assuming all probabilities in $\v p$ are of the same order,
we need the window overlap $b \gtrsim d$
to ensure no loss in accuracy \cite[\S IV.D]{Bombin2023}.
For $c$ on the other hand,
there is no such lower limit imposed
i.e.\ we could set $c =1$ with no loss in accuracy
(so long as $b \gtrsim d$).
However,
there are other reasons why $c$ cannot be too small or large:
if too small,
then many layers of computation are discarded;
if too large,
then $W$ becomes too tall
leading to slow decoding cycles
or in the case of a lookup-table decoder,
high memory requirements.
\end{rem}
\noindent
The most common choice thus is $c =b =d$
though others have been explored
e.g.\ $0 \le b <d$ \cite[\S 3]{Delfosse2023a} \cite[\S B.2]{Tan2023}
\cite[\S VII.C]{Bombin2023} \cite[\S 5.2]{Lin2025a}
or $c =1$ \cite[\S 5.1]{Das2022a_quantum_bibstyle} \cite{Huang2024,Gong2024}.
In almost all cases,
a considerable amount of computation is discarded.
This is true for the sandwich method too
as it also requires overlapping windows.
\begin{rem}\label{rem:future_boundary}
The future boundary in $W$
improves decoding accuracy \cite[\S B.3]{Tan2023}
compared to if it were absent,
as it allows the decoder to account for
small defect pairs that straddle this boundary
like the one in \cref{fig:forward_method}a.
By tentatively pairing the known defect to the future boundary,
the decoder delays the decision on how to properly annihilate it
until it discovers the other defect just above
in \cref{fig:forward_method}b.
The same occurs with another defect near the top of this window,
but this time it turns out there are no nearby defects above,
so it is instead paired to the west boundary
in \cref{fig:forward_method}c.
\end{rem}

\begin{rem}
In \emph{hard} real-time decoding,
the decoder is given a hard time limit for each decoding cycle,
which should be set to no longer than
$c$ ancilla qubit measurement intervals.
If it exceeds this limit,
it is interrupted and the logical qubit is assumed to fail.
On the other hand,
\emph{soft} real-time decoding demands less
by introducing a buffer (in the conventional sense) between
	the newest measurement data
	and the data being decoded,
in which case only \emph{an average} decoding cycle
must complete faster than $c$ intervals;
any slower then the backlog of syndrome data
grows indefinitely \cite[p~326]{Terhal2015}.
In practice,
decoding will likely be a hard real-time system
\cite[\S 2.1]{Battistel2023}.
\end{rem}
\noindent
This remark applies also to the frugal method
discussed in the next section.
We analyse both individual and mean runtime
for Snowflake in \cref{sec:throughput}.

\section{The Frugal Method}
\label{sec:the_frugal_method}
In this section we describe our
alternative to the forward method.
The frugal method can be thought of as
eagerly raising the decoding window
\emph{during} the computation of the correction,
meaning it need only be complete within the commit region.
This idea is made more precise below.
\begin{algorithm}[H]
\caption{The frugal method decoding cycle.}
\label{alg:frugal_method_decoding_cycle}
\begin{algorithmic}[1]
	\State Raise $W$ by $c$ layers,
	same as in \cref{alg:forward_method_decoding_cycle}.
	\label[step]{step:frugal_raise}
	\State Find a tentative correction $\mathbb T$ that satisfies
	\cref{cond:annihilate_all_defects_in_commit_region}.
	This condition is weaker than that in
	\cref{alg:forward_method_decoding_cycle}.
	\label[step]{step:find_tentative}
	\State Update
	$\mathbb C \gets \mathbb C \cup (\mathbb T \cap E_C)$,
	same as in \cref{alg:forward_method_decoding_cycle}.
	The difference is that $\mathbb T \cap E_B$ and
	the computation (i.e.\ the internal state of the decoder)
	associated with $B$
	is not forgotten
	but used as a starting point for,
	and further developed in,
	the next decoding cycle.
	\label[step]{step:remember_buffer_computation}
\end{algorithmic}
\end{algorithm}
\begin{cond}\label{cond:annihilate_all_defects_in_commit_region}
$\mathbb T$ annihilates $\mathbb S \cap V_C$.
\end{cond}
\begin{rem}\label{rem:frugal_method_discards}
The reuse of computation in \cref{step:remember_buffer_computation} means
zero layers of decoder computation are discarded,
regardless of $b$ or $c$ (c.f.\ \cref{rem:forward_method_discards}).
\end{rem}
\noindent
As computation is never discarded,
the power draw and thus heat generated by the device
is minimised.
This is vital for an on-chip, local decoder
integrated with the qubits,
as most qubit platforms demand low temperatures.
Compared to the forward method for $c =b$,
where 1 layer of computation is discarded per sheet of syndrome data,
the power draw is halved.

Given \cref{rem:frugal_method_discards}, for
the rest of the paper we set $c =1$ in the frugal method.
This has no effect on the decoding accuracy,
but it means the method becomes local
(provided the decoder itself is local) as,
from the perspective of the decoding window,
data shifts down by \emph{one layer}
when $W$ is raised.
The commit region $C$ is also one layer thick,
so $\mathbb C \cap E_C$
can be offloaded from the bottom of the decoding window locally.
Conceptually,
this makes the method much more granular in nature.

We note the frugal method is similar to
another recently proposed stream decoding method called \emph{fusion}~\cite[\S IV]{Wu2023a}.
Further,
the special case of the frugal method with $c =1$ is very similar to
a special case of fusion called \emph{round-wise fusion}~\cite[\S 6]{Wu2025_quantum_bibstyle},
in which the decoder updates its solution every time a new sheet of syndrome data arrives.
Like the frugal method,
fusion also avoids redundant computation,
but does so by dividing the decoding graph into windows that do not overlap,
then relying on the properties of MWPM to guarantee no loss in accuracy
when combining solutions from neighbouring windows.
Since fusion is designed specifically for use with MWPM,
it is unclear if and how much accuracy is lost when used with UF.

\section{Snowflake}\label{sec:snowflake}
In this section, we present our decoder.
\Cref{sec:snowflake_main_idea} describes Snowflake's main idea of
modifying UF for the gradual window-raising of the frugal method.
\Cref{sec:cluster_growth_schedules} discusses choices in
cluster growth that drastically affect decoding accuracy.
Snowflake uses a modified decoding window
\emph{without} a future boundary
(so in \cref{defn:decoding_window},
the highest sheet of detectors
and its incident edges is simply absent);
this will be explained in \cref{sec:constructing_the_decoding_window}.
We conclude in \cref{sec:distributed_implementation}
by describing a distributed implementation of our decoder.

In Snowflake,
we imagine flipping edges
(thus altering $\mathbb S$)
\emph{during} the computation of $\mathbb T$
in \cref{alg:frugal_method_decoding_cycle} \cref{step:find_tentative}.
Snowflake then performs \cref{step:remember_buffer_computation,step:frugal_raise}
simultaneously,
committing said flips in $E_C$ as $W$ is raised.
This has no physical consequence on the qubits,
as $\mathbb C$ is ultimately only tracked by classical electronics
in the so-called \emph{Pauli frame}
\cite{Knill2005,Riesebos2017_quantum_bibstyle}.

\subsection{Main Idea}
\label{sec:snowflake_main_idea}
The algorithm is inspired by UF
and explained by analogy to snowfall.
As snowflakes fall they gradually grow.
If two snowflakes touch they merge i.e.\ become one.
In the same way,
an active cluster nucleates at each new defect at the top of
the decoding window
and grows as it falls.
If two clusters touch they merge and become one.
\Cref{fig:snowflake} illustrates these ideas,
and below we make them more precise.

\begin{defn}\label{defn:cluster}
A \emph{cluster} $K$ is a connected subgraph of $W$
with the additional property that edges can be half-included in it.
Connectivity is defined by the fully included edges.
It \emph{touches} another cluster $L$ (or a boundary $V_\b$)
if $V_K \cap V_L \ne \varnothing$ (or $V_K \cap V_\b \ne \varnothing$).
\end{defn}
\noindent
Each edge $e \in E_W$ has a \emph{growth value},
$e.\texttt{growth} \in `{0, \frac12, 1}$
indicating how much of it is included in \emph{any} cluster.
When a cluster grows,
it does so in all possible directions
(defined by the edges of $G$)
by half an edge;
this is concretely achieved through \cref{alg:cluster_procedures} \textsc{Grow}.
\begin{algorithm}[H]
\caption{Cluster procedures in Snowflake.}
\label{alg:cluster_procedures}
\begin{algorithmic}[2]
	\Procedure{Grow}{$K$}
		\ForAll{$e \in E_W: e$ is incident to $V_K$}
			\State $e.\texttt{growth} \gets \min`{1, e.\texttt{growth} +\frac12}$
		\EndFor
	\EndProcedure
	\Statex
	\Procedure{Merge}{$K$}
		\If{$V_K$ has a boundary node $v$}
			\State $\texttt{root} :=v$
		\Else
			\State $\texttt{root} := v \in V_K$
			with highest $t$ coordinate
		\EndIf
		\For{each defect $v$ in $K$}
			\State $P :=$ a path in $K$ from $v$ to \texttt{root}
			\State flip the edges in $P$
			\label{line:push_defect_to_root}
		\EndFor
	\EndProcedure
\end{algorithmic}
\end{algorithm}
\noindent
When two clusters become connected by a fully grown edge,
they become a single cluster by \cref{defn:cluster}.
\begin{defn}\label{defn:inactive_cluster}
A cluster is \emph{inactive}
iff all its defects can be annihilated
by flipping some subset of its fully included edges;
else, it is \emph{active}.
\end{defn}
\noindent
Active clusters will eventually inactivate if they grow large enough;
this is the general idea of UF:
to minimally grow each active cluster until it inactivates.
Note the following lemma from \cite[\S A.2]{Chan2023c}.
\begin{lem}\label{lem:cluster_inactivity}
A cluster is inactive iff it has an even defect count
or touches a boundary.
\end{lem}
\noindent
Snowflake determines the first condition
by pushing all defects to a unique point in the cluster called the \emph{root},
determined in \cref{alg:cluster_procedures} \textsc{Merge}.
When two defects meet they annihilate,
so eventually zero or one defect will remain at the root,
indicating the cluster's defect count parity.
In \cref{line:push_defect_to_root},
defects are pushed to the root by flipping edges,
as per \cref{defn:push_and_annihilate}.
In this way, $\mathbb T$
	is the symmetric difference of all paths traced by defects
	and moreover is guaranteed to annihilate the defects in
	any inactive cluster.
The idea then is to inactivate all clusters
once they have fallen far enough.
Note it is not necessary for the path a defect takes to reach the root
to be of minimum length.

\subsection{Cluster Growth Schedules}
\label{sec:cluster_growth_schedules}
\begin{figure}
	\centering
	\includegraphics[width=0.48\textwidth]{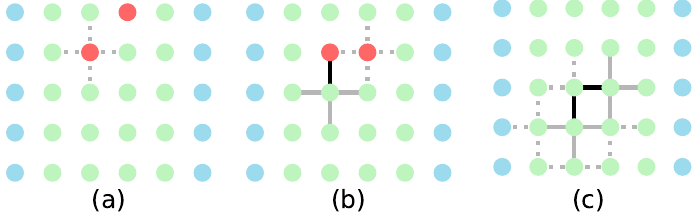}
	\caption{A concrete depiction of \cref{fig:snowflake}b
	for a distance-5 code
	to show how the 1:1 cluster growth schedule violates \cref{lem:error_covered_by_clusters}.
	Ungrown edges are invisible;
	half-grown, dotted;
	fully grown, solid.
	Edges flipped by the decoder are in black.}
	\label{fig:issue}
\end{figure}
Here we consider two different \emph{cluster growth schedules}
i.e.\ rules on how active clusters grow.
\Cref{sec:naive_growth} introduces the simplest possible one,
\emph{1:1},
and identifies an issue that hinders its accuracy.
We fix this in \cref{sec:improved_growth} with an improved schedule,
\emph{2:1}.

\subsubsection{1:1 Schedule}\label{sec:naive_growth}
In the most straightforward schedule,
Snowflake performs one growth round per decoding cycle,
as the example in \cref{fig:issue} shows.
A \emph{growth round} here comprises
	growing each active cluster once,
	redefining touching clusters,
	and pushing defects to roots.
\Cref{alg:1_1_schedule} summarises this.
\begin{algorithm}[H]
\caption{A decoding cycle of Snowflake with the 1:1 schedule,
using procedures from \cref{alg:cluster_procedures}.}
\label{alg:1_1_schedule}
\begin{algorithmic}
	\State raise $W$ by 1 layer
	\For{each active cluster $K$}
		\State \Call{Grow}{$K$}
	\EndFor
	\For{each cluster $K$}
		\State \Call{Merge}{$K$}
	\EndFor
\end{algorithmic}
\end{algorithm}
\noindent
To see the issue with this schedule,
first note the following lemma for \cite[Theorem~1]{Delfosse2021}
that is satisfied by original UF.
\begin{lem}\label{lem:error_covered_by_clusters}
Whenever a cluster grows,
the number of initial bit value-1 edges
(i.e.\ whose bit values are initially 1)
covered by the union of all clusters
increases by $\mathbin{\ge}1/2$.
\end{lem}
\noindent
This relies on the impossibility for
two active clusters to be only half an edge apart.
The theorem leads to the following for original UF.
\begin{cor}\label{cor:uf_can_correct}
Any combination of $s <d/2$ bit value-1 edges can be corrected.
\end{cor}
\begin{proof}
\Cref{lem:error_covered_by_clusters}
guarantees that after $2s$ growth rounds,
all initial bit value-1 edges are covered by the union of all clusters,
and no more growth occurs.
The largest possible cluster diameter after this many growth rounds is $2s <d$
so no logical flip is possible.
\end{proof}

For Snowflake with the 1:1 schedule,
it is possible for
two active clusters to be only half an edge apart;
see e.g.\ \cref{fig:issue}b.
Then,
during growth,
they will overgrow `into each other'
rather than stop growing as soon they touch.
This violates \cref{lem:error_covered_by_clusters}
and leads to unnecessarily large clusters.
E.g.\ assume the defect pair in \cref{fig:issue}a resulted from $s =2$ initial bit value-1 edges;
in c,
two clusters grow but only 1/2 more of the initial bit value-1 edges are covered.
The resultant cluster has diameter $9/2 >s$.

\subsubsection{2:1 Schedule}\label{sec:improved_growth}
To fix this,
we can grow clusters more cautiously.
Note in \cref{fig:issue}b
it is enough for just one of the clusters to grow
for the two to merge,
and the reason they are half an edge apart
is that one has grown an even,
and the other an odd,
number of times.
So,
define a \emph{whole (half) cluster} as one that
has grown an even (odd) number of times,
and a \emph{mixed cluster} as one that has resulted from
merging a whole with a half cluster.
Visually,
a whole (half) cluster is,
when isolated (always),
surrounded by ungrown (half-grown) edges.
A mixed cluster,
e.g.\ in \cref{fig:issue}c,
is always surrounded by both types of edge.

In this schedule,
Snowflake performs two different growth rounds per decoding cycle:
first growing the whole clusters,
then the half clusters,
as \cref{fig:fix_issue} shows.
After each round,
all clusters reevaluate their activity.
Mixed clusters need not be considered
due to the following.
\begin{prop}\label{prop:no_mixed_clusters}
Mixed clusters never occur with the 2:1 schedule.
\end{prop}
\noindent
We prove this in \cref{sec:proof_of_prop:no_mixed_clusters}.
Lastly,
for balance,
whole clusters that grow in the first round
(thus becoming half clusters)
do not grow again in the same decoding cycle.
\Cref{alg:2_1_schedule} summarises this.
\begin{algorithm}[H]
\caption{A decoding cycle of Snowflake with the 2:1 schedule,
using procedures from \cref{alg:cluster_procedures}.}
\label{alg:2_1_schedule}
\begin{algorithmic}[2]
	\State raise $W$ by 1 layer
	\For{each active whole cluster $K$}
		\State \Call{Grow}{$K$}
		\label{line:grow_all_whole_clusters}
	\EndFor
	\For{each cluster $K$}
		\State \Call{Merge}{$K$}
	\EndFor
	\For{each active half cluster $K$}
		\If{no part of $K$ grew in \cref{line:grow_all_whole_clusters}}
			\State \Call{Grow}{$K$}
		\EndIf
	\EndFor
	\For{each cluster $K$}
		\State \Call{Merge}{$K$}
	\EndFor
\end{algorithmic}
\end{algorithm}
\begin{figure}
	\centering
	\includegraphics[width=0.48\textwidth]{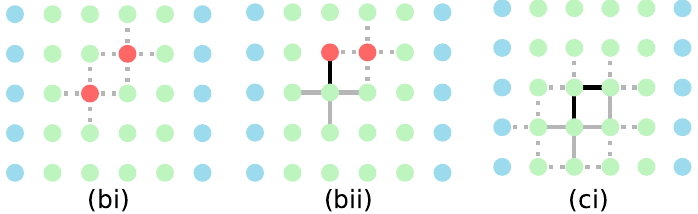}
	\caption{Continuing from \cref{fig:issue}a
	but with the 2:1 cluster growth schedule,
	which shows \cref{lem:error_covered_by_clusters} is no longer violated.
	`i' labels the first growth round of the decoding cycle;
	`ii', the second.}
	\label{fig:fix_issue}
\end{figure}

The 2:1 schedule satisfies \cref{lem:error_covered_by_clusters}
as it prevents the overgrowth issue:
if two active clusters are half an edge apart,
only one of them will grow.
E.g.\ in \cref{fig:fix_issue}ci,
only the whole cluster (lower west) grows
and the resultant half cluster has diameter $4 =2s$.
Our numerics in \Cref{sec:schedule_accuracy_comparison}
show the 2:1 schedule almost doubles the threshold of 1:1,
so for the rest of this paper,
we consider Snowflake with the 2:1 schedule.

\subsection{Constructing the Decoding Window}
\label{sec:constructing_the_decoding_window}
We now discuss the boundaries
and required height of Snowflake's decoding window.
Despite \cref{rem:future_boundary},
Snowflake uses a modified decoding window without a future boundary
because in either schedule,
clusters fall twice as fast as they grow
so this boundary would never be touched.

The height of the buffer region is lower-bounded by the following.
\begin{prop}\label{prop:min_buffer_height}
The minimum layer count of the buffer region
that guarantees \cref{cond:annihilate_all_defects_in_commit_region},
is $b_{\min} :=2\lfloor d/2 \rfloor$.
\end{prop}
\begin{proof}
Defects are annihilated by
	either meeting another defect
	or being pushed to a boundary.
So a longest-living defect is one
like that in \cref{fig:snowflake}a:
	never meeting another defect,
	and created as far as possible,
		i.e.\ $\lfloor d/2 \rfloor$ edges,
		from any boundary.
The cluster containing this defect
grows outward by half an edge per decoding cycle
so will touch a boundary after $b_{\min}$ decoding cycles,
at the $(b_{\min} + 1)^\th$ layer from the top.
So $\mathbb T$ always annihilates all defects
in this layer and below,
but not above.
\end{proof}
\noindent We set $b =b_{\min}$ for our distributed implementation.
Minimising $c$ and $b$ in the frugal method
roughly halves the height of $W$
compared to the forward method for $c =b =d$.
As \cref{fig:distributed_implementation} shows,
this more than halves the processor count of our implementation
described in the next subsection.

\subsection{Distributed Implementation}
\label{sec:distributed_implementation}
Here we describe how Snowflake can run on a distributed implementation.
We use the same local architecture of Macar and Actis in \cite{Chan2023c}:
a processor for each node in $V_W$,
a communication link for each edge
and one controller which manages the global variables of the algorithm.
For the surface code,
the processors and links form a 3D lattice
oriented so the time axis points upward,
as in \cref{fig:distributed_implementation}b.
Our decoder is pipelined
as successively lower layers in the lattice
process increasingly mature stages of decoding.

\begin{figure}
	\centering
	\includegraphics[width=0.48\textwidth]{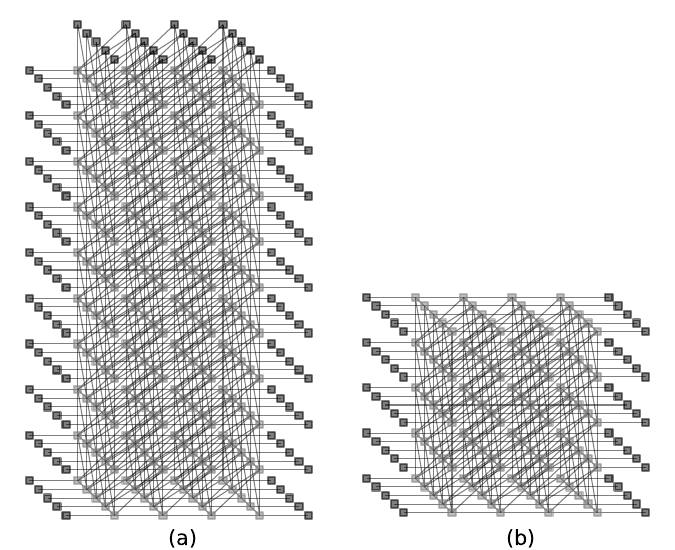}
	\caption{Distributed implementation of the decoders
	benchmarked in \cref{sec:numerics};
	the two shown here are for the distance-5 unrotated surface code.
	Each square is a classical processor;
	each line is a communication link.
	Darker squares correspond to boundary nodes in the decoding graph.
	(a) Macar \cite{Chan2023c} adapted with the forward method.
	(b) Snowflake;
	note the absence of a future boundary.}
	\label{fig:distributed_implementation}
\end{figure}

Snowflake is compatible with any choice of network
for controller--processor communication,
be it the hierarchical tree used in Macar
or the strictly local staging and signalling tree in Actis,
with synchronous or asynchronous logic.
We thus leave this as an implementation choice.

To rudimentarily measure the speed of our implementation,
we consider a simple computational model
similar to cellular automata.
Each component (i.e.\ controller or processor)
updates its state based on its previous state
and that of its neighbours.
These updates occur synchronously across the whole implementation,
allowing us to define the following.
\begin{defn}\label{defn:timestep}
A \emph{timestep} is the duration between consecutive updates
of a component's state.
\end{defn}
\noindent
The clock cycle of a hardware implementation
would roughly be proportional to a timestep.
\cref{defn:timestep} limits
how fast information propagates through the 3D lattice:
messages travel one communication link per timestep.
For example,
\cref{alg:cluster_procedures} \textsc{Grow} lasts 1 timestep,
since information need only travel from a node
to its incident edges.
On the other hand,
\textsc{Merge} lasts a variable number of timesteps
dependent on the diameter of the cluster,
because information must propagate through the whole cluster
to establish its root,
and to push defects to this root.
\Cref{sec:distributed_implementation_of_snowflake}
provides the rest of the details
of our distributed implementation.

\section{Numerics}\label{sec:numerics}

\begin{figure}
	\centering
	\includegraphics[width=0.48\textwidth]{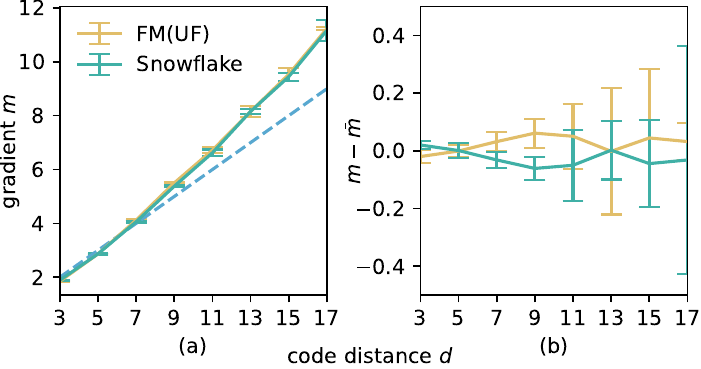}
	\caption{Errorbars show standard error.
	(a) Gradient $m$ of the lines in
	the $p \le \num{6e-3}$ region of \cref{fig:FM_UF_v_Snowflake},
	estimated using Weighted Least Squares.
	The dashed blue line shows
	the asymptotic ($p \to 0$) value $m =(d+1)/2$ from theory
	\cite[\S 5]{Fowler2012b}.
	(b) Deviation of said gradient from the mean
	$\bar m :=\frac12(m_\text{FM(UF)} +m_\text{Snowflake})$
	of both decoders.}
	\label{fig:accuracy_gradients}
\end{figure}
Here,
we numerically analyse the accuracy, throughput, and latency of Snowflake
in \cref{sec:accuracy,sec:throughput,,sec:latency}, respectively.
We wrote a Python package \cite{Chan2023a_quantum_bibstyle}
to emulate the distributed implementation of Snowflake
that follows the computational model from the previous subsection.
We use it to run the Z memory experiment
using the surface code
under the circuit-level noise model detailed in \cref{sec:noise_model}.
Specifically,
we decode the syndrome from $n \gg d$ stabiliser measurement rounds,
then count one of two features
depending on if we are benchmarking accuracy or throughput.
To benchmark latency,
we emulate and decode the end of the memory experiment.

For comparison we also benchmark FM(UF).
We fix its decoding window to $c =b =d$
as we found that smaller $c, b$ values led to both worse
accuracies and overall runtimes,
consistent with \cref{rem:small_buffer_or_commit_height}.
The distributed implementation of FM(UF) i.e.\ FM(Macar),
and of Snowflake,
is shown in \cref{fig:distributed_implementation}.
For both,
the controller--processor communication time cost is the same,
so we exclude it by assuming
the controller connects directly to each processor.

\subsection{Accuracy}\label{sec:accuracy}

We evaluate accuracy by \cref{defn:logical_error_rate}:
after decoding,
we count the number $l$ of logical flips,
identifying such paths
using the method described in \cite{Chan2024a}.
We then estimate the logical error rate $\widehat f =ld/n$.

\Cref{fig:FM_UF_v_Snowflake} compares the accuracy of
FM(UF) and Snowflake,
for practical noise levels and code distances.
Below threshold,
we assume the power law
$f/f_\th =(p/p_\th)^m$.
\Cref{fig:accuracy_gradients} shows
the subthreshold gradients $m(d)$
are indistinguishable for the two decoders
so we can characterise their accuracy difference
solely by their thresholds.

While FM(UF) has a slightly higher threshold,
Snowflake leads to
a logical qubit lifetime $1/f$ that is \num{24.9(5)}\% longer
on average over the \cref{fig:FM_UF_v_Snowflake} data.
We are unsure why Snowflake is more accurate
but suspect it may relate to
its smoother treatment of $G$.
The most obvious difference between
the outputs of UF and Snowflake
is in the relative cluster sizes
e.g.\ were UF used in \cref{fig:fix_issue},
both clusters would be the same size when they merge.
In \cref{sec:when_a_cluster_touches_both_boundaries}
we rule out another possible explanation
related to how each decoder resolves a cluster touching both boundaries.

\subsection{Throughput}\label{sec:throughput}

\begin{figure}
	\centering
	\includegraphics[width=0.48\textwidth]{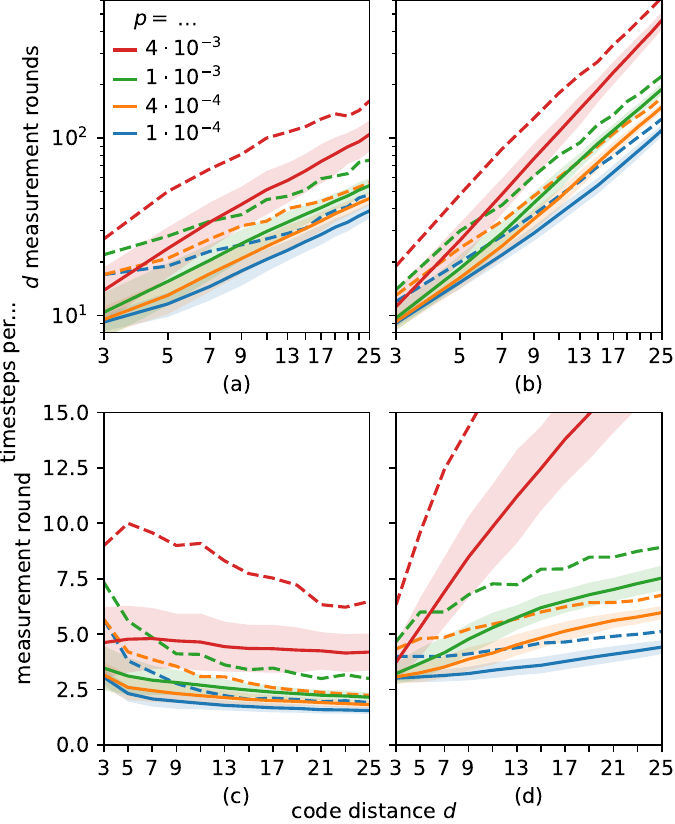}
	\caption{Decoder throughput
	against surface code distance $d$,
	for various circuit-level noise levels $p$ below threshold.
	A \emph{timestep} is defined in \cref{defn:timestep}.
	Each solid (dashed) line shows the mean ($99^\th$ percentile)
	of \num{1e3} lots of $d$ measurement rounds;
	shading shows standard deviation.
	The standard error of the mean is smaller than the solid line thickness.
	(a) Macar \cite{Chan2023c} adapted with the forward method.
	(b) Snowflake.
	(c,d) Same as a,b but rescaled.}
	\label{fig:mean_runtime}
\end{figure}

\begin{figure}
	\centering
	\includegraphics[width=0.48\textwidth]{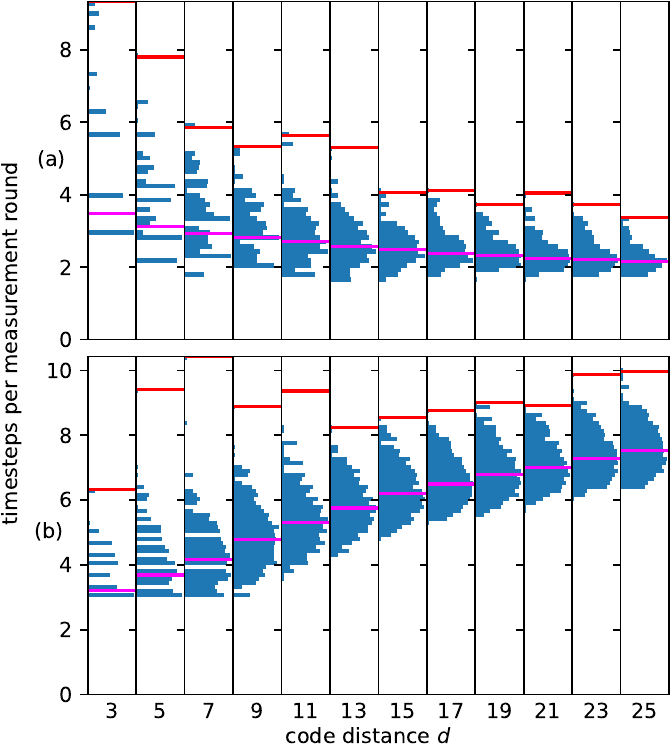}
	\caption{A more detailed view of
	the $p =\num{1e-3}$ throughput data in \cref{fig:mean_runtime}c and d,
	respectively.
	Within each histogram:
	the horizontal magenta (red) line shows the mean (maximum)
	of \num{1e3} lots of $d$ measurement rounds,
	and the $x$-axis scale is logarithmic.}
	\label{fig:runtime_distribution}
\end{figure}

For this we count the number of timesteps
(\cref{defn:timestep}) to decode.
We consider this per $d$ measurement rounds
as each logical operation lasts $\mathcal O(d)$ rounds
using lattice surgery \cite{Horsman2012},
the best known technique for universal computation
with the surface code when the quantum device has local connectivity
\cite{Fowler2019,Litinski2019}.
We also consider the timestep count \emph{per stabiliser measurement round}
because for real-time decoding,
the decoder's throughput requirement
is imposed by the quantum device's stabiliser measurement round duration \cite[\S 1]{Battistel2023}
e.g.\ \qty{\approx 1}{\micro s} for superconducting qubits
\cite[Figure~1]{Google2023} \cite[p~924]{Google2025}.

\Cref{fig:mean_runtime} compares the runtime of
FM(Macar) and Snowflake.
To adapt Macar with the forward method,
we measure its runtime to decode a window,
then add $c$ timesteps due to \cref{step:raise} of the forward method.
Using this data,
\cref{tab:gradients} estimates the scaling $m$,
assuming mean runtime $\propto d^m$.

\begin{table}[h]
\caption{Gradient $m$ of the solid lines in \cref{fig:mean_runtime}a,b
estimated using Weighted Least Squares.}
\centering
\begin{tabular}{rll}
	\toprule
	$p$	& FM(Macar)	& Snowflake \\
	\midrule
	$\num{4e-3}$	& 0.94(1)	& 1.770(4) \\
	$\num{1e-3}$	& 0.778(8)	& 1.42(1) \\
	$\num{4e-4}$	& 0.761(7)	& 1.34(2) \\
	$\num{1e-4}$	& 0.69(2)	& 1.16(2) \\
	\bottomrule
\end{tabular}
\label{tab:gradients}
\end{table}

\noindent It shows that FM(Macar)'s scaling is sublinear ($m <1$)
whereas Snowflake's is subquadratic ($m <2$).
\begin{rem}\label{rem:serialization}
We suspect this difference in mean runtime scaling is because for FM(Macar),
active clusters in the same batch all grow at the same time
regardless of their vertical coordinate.
In Snowflake,
this process is somewhat serialised:
older active clusters start growing before newer ones do,
as they enter $W$ earlier;
see e.g.\ \cref{fig:snowflake}b.
\end{rem}
\noindent
Despite this discrepancy,
subquadratic is still better than
the slightly-higher-than-cubic scaling
\cite[\S II.B]{Liyanage2023a} of original UF.

To investigate the practical regime further,
\cref{fig:runtime_distribution}
plots the distribution of runtimes for both decoders
at a practical noise level of \num{1e-3},
where Snowflake is \numrange{\approx0.9}{3.5} times slower than FM(Macar).
Neither distribution has a long tail;
such consistent runtimes are beneficial for hard real-time decoding:
the time limit can be set lower
for a given target timeout failure rate,
accelerating the quantum program
and increasing qubit coherence.
Note that Snowflake uses fewer than
half the number of processors compared to FM(Macar)
(see \cref{fig:distributed_implementation}).
Snowflake may therefore be advantageous in architectures
where decoder footprint is more of a limitation than decoding throughput.

\subsection{Latency}\label{sec:latency}

Latency is the duration
between the decoder receiving the final sheet of syndrome data
(corresponding to the end of the memory experiment)
to when the decoder commits its final tentative correction.
Low latency is crucial for gate teleportation,
where the correction to the logical observable
can no longer be tracked in the Pauli frame~\cite{Khalid2025}.

\Cref{fig:latency_distribution_1e-3} plots the distribution of latencies
for both decoders at a practical noise level of \num{1e-3}.
As before,
we adapt Macar with the forward method
by measuring its runtime to decode a window $W$
but this time add $b +c$ timesteps
because the final commit region comprises all $b +c$ layers of $W$.
The figure shows that Snowflake has a lower latency than FM(Macar);
this is mainly due to the height $d$ of Snowflake's decoding window,
which is half that of FM(Macar).

\section{Conclusion}\label{sec:conclusion}
We propose the frugal method for stream decoding
of quantum error-correcting codes
for which there exists a decoding graph embeddable in $\mathbb R^3$
without `long-range' edges.
As \cite[\S V]{Gong2024} points out,
this paradigm of
reusing decoder information from previous windows
has been explored already for \emph{classical} LDPC codes
e.g.\ in \cite[\S III.B.2]{Ali2018}.
We also design Snowflake,
a simple decoder compatible with the frugal method
and numerically benchmark its distributed implementation.
It is more accurate than original UF
and its runtime scales more favourably.
There are various avenues of further investigation
for the frugal method and Snowflake:

\begin{figure}
	\centering
	\includegraphics[width=0.48\textwidth]{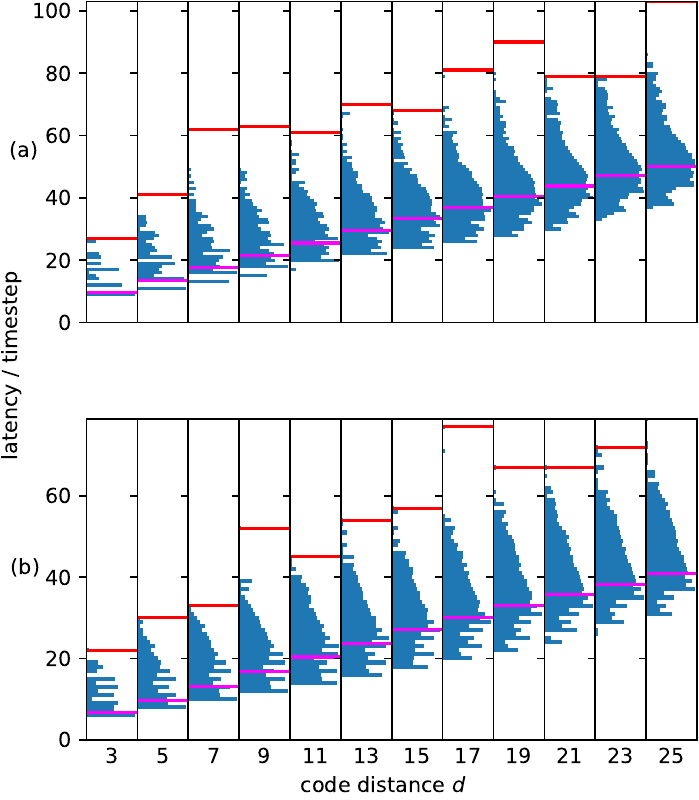}
	\caption{Decoder latency against surface code distance
	for circuit-level noise level $\num{1e-3}$.
	A \emph{timestep} is defined in \cref{defn:timestep}.
	Within each histogram:
	the horizontal magenta (red) line shows the mean (maximum)
	of \num{1e4} shots,
	and the $x$-axis scale is logarithmic.
	(a) Macar \cite{Chan2023c} adapted with the forward method.
	(b) Snowflake.}
	\label{fig:latency_distribution_1e-3}
\end{figure}

\paragraph{2D Architecture}
The serialisation in the vertical direction
mentioned in \cref{rem:serialization}
suggests the viability of an alternative architecture:
`flattening' the 3D lattice down onto a 2D grid
so there is one processor
for each \emph{set of nodes} of the same $(x, y)$ coordinate
i.e.\ each ancilla qubit.
This comes at the cost of
serialising the operations within each set,
but if these operations are inherently serial,
the cost could be marginal.
There are two benefits:
	first,
		it is more practical
		as a 2D chip is far easier to build than a 3D lattice of processors;
	second,
		messages passed vertically would be instantaneous within a timestep.
Future work would involve emulation and numerical analysis of this flattened implementation.

\paragraph{Asynchronous Logic} As discussed in \cite[\S 6]{Chan2023c},
controller--processor communication could run by
asynchronous (instead of synchronous) logic,
which would speed up Snowflake.
One could explore which other processes in Snowflake
could be sped up by asynchronous logic.

\paragraph{Weighted Edges} Clusters could grow along edges at different speeds
to account for their different probabilities in $\v p$,
which would improve accuracy \cite{Huang2020}.
These speeds could even be configurable in the hardware
to fine-tune the noise model during a calibration stage;
this however would alter \cref{prop:min_buffer_height},
requiring $b_{\min}$ to also be configurable,
which is nontrivial.

\paragraph{Computation} Recently,
the sandwich method has been generalised
	from quantum memory
	to lattice surgery-based computation
\cite{Gidney2022,Bombin2023,Lin2025a}.
We propose an alternative approach to decoding lattice surgery
that exploits the granular and local nature of Snowflake.
Imagine the decoder as a homogeneous substrate
positioned just above
the entire 2D lattice of qubits.
Regions of this substrate are turned on (to decode) wherever
the below qubits are involved in surgery.
Snowflake would work best with so-called
\emph{twist-free} lattice surgery \cite[\S IV]{Chamberland2022}
which minimises the use of nonlocal stabilisers.
These elongated stabilisers
introduce in $G$ additional detectors \cite[Appendix~A.3]{Geher2024}
and accounting for them in our scheme would be nontrivial.
Some details of the distributed implementation in \cref{sec:setup}
would need to be generalised.
E.g.\ since boundaries vary during lattice surgery,
the \verb|ID| of each processor
would no longer be a constant integer;
rather, a combination of
	a constant 3D coordinate
	and a (relative) variable boolean indicating whether
		the processor represents a boundary node.
This preserves \cref{rem:boundaries_lower_ID_than_detectors}.
For controller--processor communication,
a hierarchical tree
like that used in Macar
would be best,
as communication cost would then scale logarithmically
with the length of the qubit lattice,
which may be much larger than $d$.

\begin{acknowledgments}
I thank
	Simon Benjamin,
	Matt Cassie,
	and Armands Strikis
for useful discussions,
and
	Luka Skoric,
	Alexandru Paler,
	and Timo Hillmann
for helpful comments.
I acknowledge the use of
the University of Oxford Advanced Research Computing
(ARC)
facility~\cite{Richards2015_quantum_bibstyle} in carrying out this work
and specifically the facilities made available
from the EPSRC QCS Hub grant
(agreement No.\ EP/T001062/1).
I also acknowledge support from
an EPSRC DTP studentship,
two EPSRC projects
RoaRQ (EP/W032635/1)
and SEEQA (EP/Y004655/1),
and JST ASPIRE Japan Grant Number JPMJAP2319.
\end{acknowledgments}

\section*{Competing Interests}
The author declares a relevant patent application pending:
PCT/EP2025/063485,
published as WO2025238191A1 on 2025-11-20.

\bibliographystyle{quantum}
\bibliography{tchbib}

\appendix

\section{Noise Model}
\label{sec:noise_model}
\begin{figure}[H]
	\centering
	\begin{tabular}{cc}
		\begin{tikzpicture}
			\begin{yquant}
			qubit {\textsc{n}} d[1];
			qubit {\textsc{w}} d[+1];
			qubit {\textsc{e}} d[+1];
			qubit {\textsc{s}} d[+1];
			qubit {\textsc{a}} a;
			cnot a | d[0];
			cnot a | d[1];
			cnot a | d[2];
			cnot a | d[3];
			[type=qubit] measure {\textsc{z}} a;
			\end{yquant}
		\end{tikzpicture}
		&
		\begin{tikzpicture}
			\begin{yquant}
			qubit {\textsc{n}} d[1];
			qubit {\textsc{w}} d[+1];
			qubit {\textsc{e}} d[+1];
			qubit {\textsc{s}} d[+1];
			qubit {\textsc{a}} a;
			cnot d[0] | a;
			cnot d[1] | a;
			cnot d[2] | a;
			cnot d[3] | a;
			[type=qubit] measure {\textsc{x}} a;
			\end{yquant}
		\end{tikzpicture}
		\\
		(a) & (b)
	\end{tabular}
	\caption{The (a) Z- and (b) X-stabiliser measurement circuits for the surface code.
	\textsc{n, w, e, s} are the data qubits
	north,
	west,
	east,
	and south of ancilla qubit \textsc{a}, respectively.
	Each data qubit interacts with a different ancilla qubit
	every step for the first four steps
	but only one interaction is shown.
	The initial state of \textsc{a} is
	the measurement outcome from the previous round
	i.e.\ $\ket{0}$ or $\ket{1}$ for (a),
	and $\ket{+}$ or $\ket{-}$ for (b).}
	\label{fig:5_step}
\end{figure}
\begin{table}[h!]
	\centering
	\begin{tabularx}{\linewidth}{rX}
		\toprule
		Probability & Error Process \\
		\midrule
		$p_\M$ & qubit measurement reports the opposite outcome \\
		$p_1$ & qubit idle is followed by
		an error drawn randomly from
		$`{\hat \X, \hat \Y, \hat \Z}$ \\
		$p_2$ & $\C \hat \X$ is followed by
		an error drawn randomly from
		$`{\hat\I, \hat \X, \hat \Y, \hat \Z}^{\otimes 2}
		\setminus `{\hat\I{}^{\otimes 2}}$ \\
		\bottomrule
	\end{tabularx}
	\caption{Error processes for the circuit-level depolarising noise model.
	Each process can occur once in each of the 5 steps in \cref{fig:5_step}.}
	\label{tab:circuit_level_depolarising_noise}
\end{table}
\noindent
Our numerics were performed using the same noise model as in \cite[\S B]{Chan2023c}.
Namely,
we assume the stabiliser measurement circuits in \cref{fig:5_step}
and apply the error processes in \cref{tab:circuit_level_depolarising_noise},
using the balanced parametrisation from \cite[p~1]{Wang2011}:
\begin{equation}
(p_\M, p_1, p_2) =`\Big(\frac23 \frac45, \frac45, 1) p,
\end{equation}
where $p$ is the noise level.

\section{Proof of \cref{prop:no_mixed_clusters}}
\label{sec:proof_of_prop:no_mixed_clusters}
We first define three functions.
The \emph{wholeness} of a non-mixed cluster $K$ is
\begin{equation}
\wh K :=\begin{cases*}
1	&$K$ a whole cluster, \\
\frac12	&$K$ a half cluster.
\end{cases*}
\end{equation}
The \emph{ungrown length} of a path $P$ in $W$ that comprises
$f$ fully grown edges,
$h$ half-grown edges,
and $z$ ungrown edges,
is
\begin{equation}
l(P) :=\tfrac12 h +z.
\end{equation}
The \emph{ungrown distance} between clusters $K, L$
is the minimum ungrown length of a path
	from a node in $K$
	to a node in $L$:
\begin{equation}
d(K, L) :=\min`{l(uv, \dots, wx):
\text{$ux \in V_K \times V_L$}}.
\end{equation}
Now consider the following.
\begin{lem}\label{lem:whole_half_distance}
After each growth round of Snowflake with the 2:1 schedule,
we have
\begin{equation}\label{eq:whole_half_distance}
d(K, L) \in \mathbb Z +\wh K +\wh L
\end{equation}
for any two non-mixed clusters $K, L$.
\end{lem}
\begin{proof}
By induction on the number $r$ of growth rounds elapsed.

Consider base case $r =0$.
All $e \in E_W$ are ungrown
so $l(P) \in \mathbb Z$ for any path $P$ in $W$,
and so $d(K, L) \in \mathbb Z$.
\emph{All} clusters are whole
so $\wh K +\wh L =2$,
and so \cref{eq:whole_half_distance} is true for $r =0$.

Now the inductive step;
assume true after $r$ growth rounds
and consider growth round $r +1$.
If\dots
\begin{enumerate}
	\item neither $K$ nor $L$ grows,
	all terms in \cref{eq:whole_half_distance} are unchanged;
	\item $K$ grows but not $L$,
	$d(K, L)$ decreases by $\frac12$
	and $\wh K$ changes by $\frac12$;
	\item both $K$ and $L$ grow,
	we must have $\wh K +\wh L \in `{1, 2}$,
	so $K$ and $L$ cannot be half an edge apart before growth,
	so $d(K, L)$ decreases by 1.
	Also,
	$\wh K$ and $\wh L$ each change by $\frac12$.
\end{enumerate}
In each case,
\cref{eq:whole_half_distance} remains true.
\end{proof}
\noindent
Now we can prove the desired statement.
\begin{proof}[Proof of \cref{prop:no_mixed_clusters}]
There are no mixed clusters initially.
By \cref{lem:whole_half_distance},
a whole cluster never merges with a half cluster,
so mixed clusters never form.
\end{proof}

\section{Schedule Accuracy Comparison}
\label{sec:schedule_accuracy_comparison}

\Cref{sec:cluster_growth_schedules}
discussed 1:1 and 2:1
as two possible cluster growth schedules for Snowflake.
\Cref{fig:schedule_accuracy_comparison} shows
how Snowflake performs with each schedule:
2:1 almost doubles the threshold $p_\th$ of 1:1.

\begin{figure}[H]
	\centering
	\includegraphics[width=0.48\textwidth]{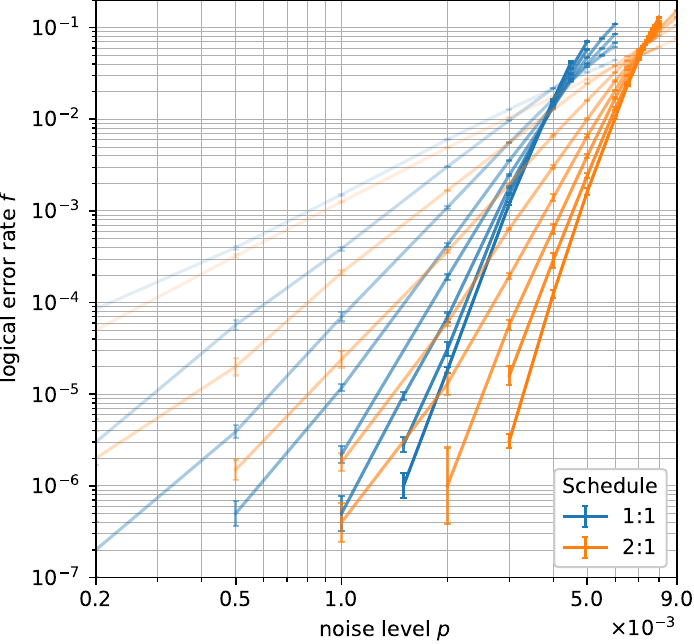}
	\caption{Same as \cref{fig:FM_UF_v_Snowflake}
	but the blue lines instead show data for
	Snowflake with the 1:1 cluster growth schedule.
	The coordinate $(p, f)_\th$ of the threshold is around
	$(0.375, 9)\cdot\num{1e-2}$ for 1:1, and
	$(0.735, 6.8)\cdot\num{1e-2}$ for 2:1.}
	\label{fig:schedule_accuracy_comparison}
\end{figure}

\section{Distributed Implementation of Snowflake}
\label{sec:distributed_implementation_of_snowflake}
In this section we refer to
processors as nodes and
communication links as edges.
The controller can broadcast the same boolean to,
and receive (without knowing the sender) booleans from,
every node.

In \cref{sec:setup} we list the variables each component stores
and in \cref{sec:procedures} we explain how they are manipulated
in order to effect Snowflake.
\Cref{sec:processes_common_to_both_merging_stages,sec:processes_unique_to_merging_whole}
explain the variable-duration stages in more detail.
\Cref{sec:after_data_qubit_measurement} describes how we minimise decoding latency
after all data qubits are measured out.

\subsection{Setup}\label{sec:setup}
Each node $v \in V_W$ is assigned a unique integer \verb|ID|
and stores the following variables:

\begin{enumerate}
	\item \verb|active| is a boolean indicating
	whether $v$ is in an active cluster.
	If $v.\verb|active| =\verb|true|$ we say $v$ is active.
	\item \texttt{whole} is a boolean indicating
	whether $v$ is in a whole cluster.
	If $v.\texttt{whole} =\texttt{true}$ we say $v$ is whole.
	\item \verb|CID| is an integer indicating
	the cluster that $v$ belongs to.
	Clusters are identified by the lowest \verb|ID| of all its nodes.
	The node of this \verb|ID| is the root.
	\verb|CID| can also adopt the value \verb|reset|
	to indicate $v$ should \emph{unroot}
	(defined in \cref{sec:unrooting})
	in the next timestep.
	\item \verb|defect| is a boolean indicating whether $v$ has a defect.
	This can be thought of as a particle which is passed between nodes.
	When a node receives two defects in a timestep,
	they annihilate.
	\item \verb|pointer| is the direction $v$ should push the defect
	(by flipping the corresponding edge).
	Possible values (under circuit-level noise) are
		\texttt{C},
		\texttt{N},
		\texttt{W},
		\texttt{E},
		\texttt{S},
		\texttt{D},
		\texttt{U},
		\texttt{NU},
		\texttt{WD},
		\texttt{EU},
		\texttt{SD},
		\texttt{NWD},
		\texttt{SEU},
	representing respectively
		centre,
		north,
		west,
		east,
		south,
		down,
		up,
	and combinations thereof.
	Following the pointers starting from $v$ should eventually lead to
	the root of the cluster that $v$ is in.
	Only roots have $\verb|pointer| =\verb|C|$
	as they do not push defects but accumulate them.
	\item \texttt{grown} is a boolean indicating
	whether the cluster that $v$ is in
	has grown in the first growth round of the current decoding cycle.
	\item \verb|unrooted| is a boolean indicating whether
	$v$ has unrooted in the current decoding cycle.
	\item \verb|busy| is a boolean indicating whether $v$ is busy.
	\item \verb|stage| is the current stage of Snowflake.
	This follows the flowchart in \cref{fig:stage_flowchart}.
\end{enumerate}

\begin{figure}[H]
	\centering
	\begin{tikzpicture}[
		stage/.style={
			draw=black,
			text centered,
			text depth=.25ex,
			text height=1.5ex
		},
		arrow/.style={-{Stealth[length=2mm]}},
	]
	% nodes
	\node[stage] (drop) {\texttt{drop}};
	\node[stage] (grow_whole) [right=of drop] {\texttt{grow\_whole}};
	\node[stage] (merging_whole) [right=of grow_whole] {\texttt{merging\_whole}};
	\node[stage] (grow_half) [below=of grow_whole] {\texttt{grow\_half}};
	\node[stage] (merging_half) [below=of merging_whole] {\texttt{merging\_half}};
	% lines
	\draw[arrow] (drop.east) -- (grow_whole.west);
	\draw[arrow] (grow_whole.east) -- (merging_whole.west);
	\draw[arrow] ([xshift=2mm]merging_whole.south) -- ++(0,-4mm) -- ++(15mm,0) |- (merging_whole.east);
	\draw[arrow,dashed] (merging_whole.south west) -- (grow_half.north east);
	\draw[arrow] (grow_half.east) -- (merging_half.west);
	\draw[arrow] ([xshift=2mm]merging_half.south) -- ++(0,-4mm) -- ++(15mm,0) |- (merging_half.east);
	\draw[arrow,dashed] ([xshift=-2mm]merging_half.south) -- ++(0,-4mm) -| (drop.south);
	\end{tikzpicture}
	\caption{Flowchart for the stages of Snowflake:
	\texttt{drop}, \texttt{grow\_whole}, and \texttt{grow\_half} each last one timestep
	whereas the other two last a variable number,
	hence their loops.
	The dashed (solid) arrows indicate advances by
	the global controller (each node locally).
	Snowflake cycles through all five stages
	once per decoding cycle.}
	\label{fig:stage_flowchart}
\end{figure}

\begin{figure*}
	\centering
	\includegraphics[width=0.7\textwidth]{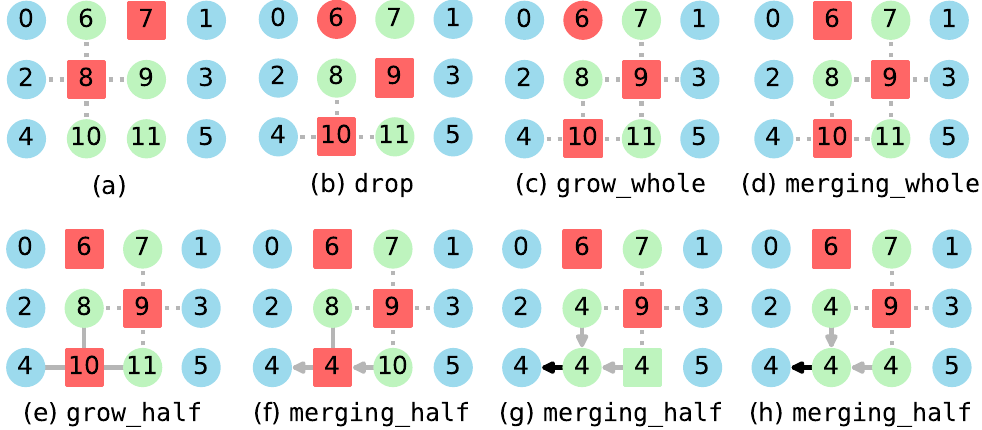}
	\caption{A decoding cycle example for a distance-3 code.
	Edges are shown as in \cref{fig:issue}.
	Active nodes are squares;
	inactive nodes, circles.
	\texttt{CID} is shown as a label.
	Pointers are shown by arrows on edges.
	(a) Each node is in its own cluster
	so has $\texttt{CID} =\texttt{ID}$.
	Two clusters are active.
	(b) \emph{All} relative variables fall by one layer,
	including \texttt{CID} which is adjusted to refer to the node below.
	A new defect appears at the top.
	(c) Node 9 grows its incident edges by $\frac12$
	as it is active and whole.
	(d) Node 6 recognises it is a root with a defect,
	so becomes active.
	(e) Node 10 grows its incident edges by $\frac12$
	as it is active and half.
	Nodes 10 and 11 now each see a lower \texttt{CID} than their own
	when they look at their neighbours along fully grown edges.
	(f) They adopt that \texttt{CID}
	and point toward that neighbour;
	node 8 thus sees a lower \texttt{CID} than its own.
	(g) A defect is pushed along a pointer to the root
	which flips an edge.
	Node 8 updates its \texttt{CID} and \texttt{pointer},
	so now the 4-node cluster is fully flooded but not yet fully synced.
	(h) Throughout this stage,
	each non-root node was adopting the \texttt{active} variable
	of the node it was pointing to.
	This activity propagation now completes,
	so the 4-node cluster is fully synced.}
	\label{fig:drop_grow_merging}
\end{figure*}

\noindent
We refer to the first five variables as \emph{relative}
and the last four as \emph{absolute}.
Each edge $e \in E_W$ has the following variables:
\begin{enumerate}
	\item \verb|growth| is its growth value
	which is in $`{0, \frac12, 1}$.
	\item \verb|correction| is a boolean indicating
	if $e \in \mathbb T$.
\end{enumerate}
These are stored as additional relative variables
at the endpoint with the lower \verb|ID|.
It helps to define the \emph{owned edges} as those whose variables a node stores:
\begin{equation}
v.\texttt{owned} :=`{uv: uv \in E_W~\text{and}~v.\texttt{ID} < u.\texttt{ID}}.
\end{equation}

\subsection{Procedures}\label{sec:procedures}

Though the pseudocode in this subsection explains exactly what each component does,
it is difficult to infer from this the \emph{emergent effect}.
We describe the latter in the accompanying prose.

\begin{algorithm}[H]
\caption{Run by controller every timestep.}\label{alg:controller_advance}
\begin{algorithmic}
	\If{not (waiting or any $v.\texttt{busy}$)}
		\For{$v \in V_W$}
			\State $v.\texttt{stage} \gets \text{next stage}$
		\EndFor
		\If{$v.\texttt{stage}$ was \texttt{merging\_whole}}
			\State wait 1 timestep
		\Else \Comment{$v.\texttt{stage}$ was \texttt{merging\_half}}
			\State wait 2 timesteps
		\EndIf
	\EndIf
\end{algorithmic}
\end{algorithm}

In each timestep the controller runs \cref{alg:controller_advance} which,
during either merging stage,
tells all nodes to advance to the next stage
if none of them are busy.
This is the only time the controller must communicate with the nodes,
as the merging stages are the only ones of variable duration.

\begin{algorithm}[H]
\caption{Run by node $v$ every timestep.}\label{alg:node_advance}
\begin{algorithmic}
	\State call the procedure in \cref{alg:procedures} matching $v.\texttt{stage}$
	\If{$v.\texttt{stage} \not\ni \texttt{merging}$}
		\State $v.\texttt{stage} \gets$ next stage
	\EndIf
\end{algorithmic}
\end{algorithm}

\begin{figure*}
\begin{minipage}{\linewidth}
\begin{algorithm}[H]
\caption{The five procedures that can be called in \cref{alg:node_advance}.}\label{alg:procedures}
\begin{algorithmic}[2]
	\Procedure{Drop}{$v$}
		\State $v.\texttt{grown} \gets \texttt{false}$ \label{line:reset_grown}
		\For{$e \in v.\texttt{owned} \cap \text{commit region}$}
			\If{$e.\texttt{correction}$}
			\Comment{Commit $e$.}
				\State $\mathbb C \gets \mathbb C \sd `{e}$
			\EndIf
		\EndFor
		\If{$\exists$ node $u$ below $v$}
		\Comment{Shift relative data down one layer.}
			\State $v.\texttt{CID} \gets$ \texttt{ID} of node below current root
			\State send all relative variables to $u$
		\EndIf
		\If{$v \in$ top sheet of $W$}
		\Comment{Inherit the new measurement round results.}
			\State $v.\texttt{defect} \gets$ whether
			the measurement of the corresponding ancilla qubit
			differs from the last round \label{line:set_defect}
			\State $v.\texttt{active} \gets \texttt{false}$ \label{line:reset_active}
		\EndIf
	\EndProcedure
	\State
	\Procedure{GrowWhole}{$v$}
		\If{$v.\texttt{active}$ and $v.\texttt{whole}$}
			\State \Call{Grow}{$v$}
			\State $v.\texttt{grown} \gets \texttt{true}$
		\EndIf
		\If{$v \in$ bottom sheet of $W$ and $v.\texttt{pointer} \ni\texttt{D}$}
		\Comment{Start unrooting $v$.}
		\label{line:pointer_D}
			\State $v.\texttt{CID} \gets \texttt{reset}$
			\State $v.\texttt{pointer} \gets \texttt{C}$
		\EndIf
	\EndProcedure
	\State
	\Procedure{GrowHalf}{$v$}
		\State $v.\texttt{unrooted} \gets \texttt{false}$ \label{line:reset_unrooted}
		\If{$v.\texttt{active}$ and not $v.\texttt{whole}$ and not $v.\texttt{grown}$}
			\State \Call{Grow}{$v$}
		\EndIf
	\EndProcedure
	\State
	\Procedure{MergingWhole}{$v$}
		\State $v.\texttt{busy} \gets \texttt{false}$
		\State \Call{Syncing}{$v$}
		\State \Call{FloodingWhole}{$v$}
	\EndProcedure
	\State
	\Procedure{MergingHalf}{$v$}
		\State $v.\texttt{busy} \gets \texttt{false}$
		\State \Call{Syncing}{$v$}
		\State \Call{FloodingHalf}{$v$}
	\EndProcedure
\end{algorithmic}
\end{algorithm}
\end{minipage}
\end{figure*}

\begin{figure*}
	\centering
	\includegraphics[width=\textwidth]{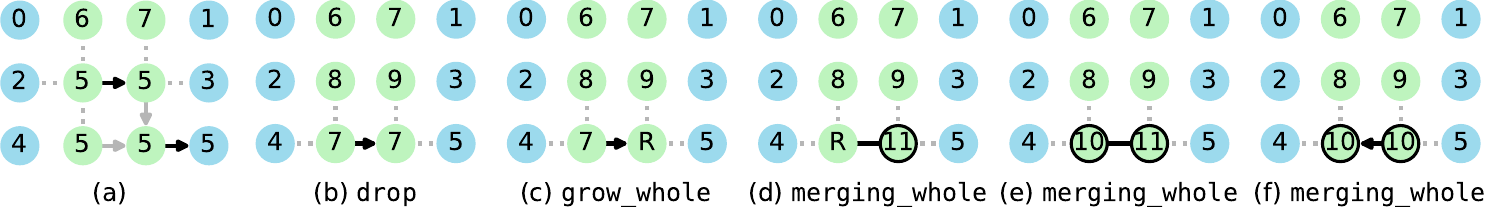}
	\caption{A decoding cycle which demonstrates unrooting.
	Relative variables are shown as in \cref{fig:drop_grow_merging}.
	\texttt{CID} value \texttt{reset} is labelled `\textsf{R}'.
	Nodes with \texttt{unrooted} as \texttt{true} are outlined in black.
	(a) There is a 5-node cluster.
	(b) A downward pointer is dropped to the bottom east detector.
	(c) Thus, this detector starts unrooting.
	(d) The bottom west detector sees this signal to unroot
	so starts unrooting.
	The bottom east detector finishes unrooting.
	(e) The bottom west detector finishes unrooting.
	(f) \texttt{merging\_whole} completes.}
	\label{fig:unrooting}
\end{figure*}

At the same time,
each node $v$ runs \cref{alg:node_advance},
calling procedures from \cref{alg:procedures}.
Stage \texttt{drop}
	adds to $\mathbb C$
	and raises $W$ by one layer,
as in \cref{fig:drop_grow_merging}b.
Each node $v$ in the top sheet of $W$ starts in its own cluster,
thus sets its relative variables to reflect this (\cref{line:set_defect,line:reset_active}).
$v.\texttt{whole}$,
$v.\texttt{CID}$,
and $v.\texttt{pointer}$
need not be explicitly reset
as they never change from their correct values of
\texttt{true},
$v.\texttt{ID}$,
and \texttt{C},
respectively.

Stage \texttt{grow\_whole} grows all edges around
each active whole cluster by $\frac12$,
as in \cref{fig:drop_grow_merging}c.
Stage \texttt{grow\_half} does the same
but for active half clusters that have not already grown
in the current decoding cycle,
as in \cref{fig:drop_grow_merging}e.

\Cref{fig:drop_grow_merging}d shows a merging example
that lasts just one timestep,
and f--h shows one that needs three timesteps.
The \texttt{merging\_whole} stage is the same as
\texttt{merging\_half} but with two extra processes,
so we first explain the latter:
it comprises two processes,
\emph{syncing} and \emph{flooding},
which occur independently.
We thus split \textsc{MergingHalf} into two subroutines,
given in \cref{alg:subroutines},
effecting each process.

\subsection{Processes Common to Both Merging Stages}
\label{sec:processes_common_to_both_merging_stages}
In flooding,
newly touching clusters adopt the lowest \verb|CID| among them.
This \verb|CID| propagates as a flood
starting from where they touched:
each node looks at its neighbours along fully grown edges.
If it sees a lower \verb|CID| than its own,
it adopts that \verb|CID|
and points toward that neighbour
(\crefrange{line:see_lower_CID}{line:match_CID}).
After this process,
all nodes in a given cluster will have the same \verb|CID|,
and following the pointers
will always lead to the root promised by \verb|CID|.

\begin{rem}\label{rem:boundaries_lower_ID_than_detectors}
Crucial to the algorithm is that
boundary nodes are of lower \verb|ID| than detectors.
Each root can thus determine if its cluster touches a boundary
by checking if itself is a boundary node.
\end{rem}

In syncing,
all defects are pushed along pointers toward roots,
one edge per timestep (\crefrange{line:if_x}{line:flip_uv_correction}).
Eventually,
each root will know using \cref{lem:cluster_inactivity}
whether its cluster is active (\cref{line:active_gets_x}).
At the same time,
this activity knowledge propagates
in the opposite direction
from root to the rest of the cluster (\cref{line:get_active_from_pointee}).
After this process,
a node is in an active cluster iff it is active.

\begin{algorithm}[H]
\caption{Subroutines for \cref{alg:procedures}.}\label{alg:subroutines}
\begin{algorithmic}[2]
	\Procedure{Grow}{$v$}
		\For{each edge $e$ incident to $v$}
			\State $e.\texttt{growth} \gets \min`{1, e.\texttt{growth} +\frac12}$
		\EndFor
		\State flip $v.\texttt{whole}$
	\EndProcedure
	\State
	\Procedure{Syncing}{$v$}
		\State $x :=(v~\text{is a detector})$ and $v.\texttt{defect}$
		\If{$v.\texttt{pointer} =\texttt{C}$}
		\Comment{$v$ is a root.}
			\State $v.\texttt{active} \gets x$ \label{line:active_gets_x}
		\Else
			\State $u :=$ node toward $v.\texttt{pointer}$
			\State $v.\texttt{active} \gets u.\texttt{active}$ \label{line:get_active_from_pointee}
			\If{$x$} \label{line:if_x}
				\State $v.\texttt{busy} \gets \texttt{true}$
				\State flip $uv.\texttt{correction}$ to push defect to $u$\label{line:flip_uv_correction}
			\EndIf
		\EndIf
		\If{$v.\texttt{active}$ changed}
			\State $v.\texttt{busy} \gets \texttt{true}$
		\EndIf
	\EndProcedure
	\State
	\Procedure{FloodingWhole}{$v$}
		\If{$v.\texttt{CID} =\texttt{reset}$}
		\Comment{Finish unrooting $v$.}
			\State $v.\texttt{busy} \gets \texttt{true}$
			\State $v.\texttt{CID} \gets v.\texttt{ID}$
			\State $v.\texttt{unrooted} \gets \texttt{true}$
		\Else
			\ForAll{$u: uv \in E_W \wedge uv.\texttt{growth} =1$} \label{line:access}
				\If{$u.\texttt{CID} =\texttt{reset}$}
					\If{not $v.\texttt{unrooted}$} \label{line:if_not_unrooted}
					\LComment{Start unrooting $v$.}
						\State $v.\texttt{busy} \gets \texttt{true}$
						\State $v.\texttt{CID} \gets \texttt{reset}$
						\State $v.\texttt{pointer} \gets \texttt{C}$
						\State break \label{line:break}
					\EndIf
				\Else
					\State \Call{Compare}{$u, v$}
				\EndIf
				\If{$u.\texttt{grown}$ and not $v.\texttt{grown}$} \label{line:see_grown_true}
					\State $v.\texttt{busy} \gets \texttt{true}$
					\State $v.\texttt{grown} \gets \texttt{true}$ \label{line:set_grown_to_true}
				\EndIf
			\EndFor
		\EndIf
	\EndProcedure
	\State
	\Procedure{FloodingHalf}{$v$}
		\ForAll{$u: uv \in E_W \wedge uv.\texttt{growth} =1$}
			\State \Call{Compare}{$u, v$}
		\EndFor
	\EndProcedure
	\State
	\Procedure{Compare}{$u, v$}
		\If{$u.\texttt{CID} <v.\texttt{CID}$} \label{line:see_lower_CID}
			\State $v.\texttt{busy} \gets \texttt{true}$
			\State $v.\texttt{pointer} \gets$ toward $u$
			\State $v.\texttt{CID} \gets u.\texttt{CID}$ \label{line:match_CID}
		\EndIf
	\EndProcedure
\end{algorithmic}
\end{algorithm}

Note that swapping the order of the two subroutines in \textsc{MergingHalf}
would speed up the stage,
as pointers would be updated \emph{before}
being used to push defects
and backpropagate activity.
However,
the practicality of this depends on the hardware
so in this paper we assume the worse case
and perform \textsc{Syncing} using pointers from the \emph{previous} timestep.

\subsection{Processes Unique to \texttt{merging\_whole}}
\label{sec:processes_unique_to_merging_whole}
This subsection will refer to both \cref{alg:procedures,alg:subroutines}.

\subsubsection{Tracking Which Clusters Have Grown}
The variable \texttt{grown} floods
in the same fashion as \texttt{cid}
but with no connection to \texttt{pointer}
(\cref{alg:subroutines} \crefrange{line:see_grown_true}{line:set_grown_to_true}).
This ensures all nodes in a given cluster
will have the same boolean for \texttt{grown}.
The memory of whether a cluster has grown
resets after each decoding cycle (\cref{alg:procedures} \cref{line:reset_grown}).

\subsubsection{Unrooting}\label{sec:unrooting}
\emph{Unrooting} is a process unique to Snowflake
and is needed when
a downward pointer is dropped to the bottom sheet of $W$
(can happen while a cluster falls out of view).
This breaks the pointer tree structure established in the previous decoding cycle,
as $\exists$ nonempty $S \subset V_W$
from which following the pointers will lead past the bottom of $W$.
To fix it,
each node in $S$ resets its \verb|CID| and \verb|pointer|
so that the structure can be rebuilt from scratch,
via further \texttt{merging\_whole} timesteps.
\Cref{fig:unrooting} shows an example.

The breaking point is found immediately after \texttt{drop}
(\cref{alg:procedures} \cref{line:pointer_D})
whence the signal to unroot,
$\verb|CID| =\verb|reset|$,
propagates during \texttt{merging\_whole} as a flood through the cluster;
each node unroots itself as soon as it receives this signal
(\cref{alg:subroutines} \crefrange{line:access}{line:break}).
It takes 2 timesteps for a node to unroot,
after which it cannot unroot again in the same decoding cycle
(\cref{alg:subroutines} \cref{line:if_not_unrooted}).
This ensures the flood does not backpropagate.
The memory of whether a node has unrooted
resets after this stage finishes (\cref{alg:procedures} \cref{line:reset_unrooted}).

The higher up a node is in $W$,
the lower its \verb|ID|;
this is for two reasons.
First,
it makes unrooting less likely
as pointer paths tend to lead upward rather than downward.
Second,
it ensures defects in active clusters
are pushed to a highest point
so they can annihilate with defects in newer, smaller clusters
as early as possible.

\subsection{After Data Qubit Measurement}
\label{sec:after_data_qubit_measurement}
As \cref{sec:latency} describes,
data qubit measurement (e.g.\ at the end of a memory experiment)
means there will be a final sheet of syndrome data.
In the worst case,
Snowflake needs $b +1$ decoding cycles
between receiving this final sheet
and committing its final tentative correction,
but on average,
most of these growing and merging stages are unnecessary
(when there are no more defects in the decoding window).
To minimise latency,
the distributed implementation bypasses these unnecessary stages by,
upon receiving the final sheet of syndrome data,
modifying \cref{alg:controller_advance} as follows:
the controller additionally checks after each decoding cycle
for any defects in the decoding window;
if there are none,
the controller repeatedly sets the stage of all nodes to \texttt{drop}
so that only \texttt{drop} stages are performed.

\section{When a Cluster Touches Both Boundaries}
\label{sec:when_a_cluster_touches_both_boundaries}
\begin{figure}[H]
	\centering
	\includegraphics[width=0.3\textwidth]{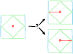}
	\caption{A cluster that touches both west and east boundaries.
	If it contains an odd number of defects,
	the correction must contain a path that connects to one of them.}
	\label{fig:arbitrary_choice}
\end{figure}
In either UF or Snowflake,
whenever a cluster has grown large enough to touch both the west and east boundaries,
there is an arbitrary choice of which boundary to push defects to,
as \cref{fig:arbitrary_choice} illustrates.
Whether the decoder chooses the same boundary or a random boundary every time
may impact the overall logical error rate.
Snowflake always chooses the west boundary,
whereas our implementation of FM(UF) may not necessarily.
To see if this explains the accuracy discrepancy between FM(UF) and Snowflake,
we reran emulation with a version of FM(UF) that always chooses the west boundary.
As \cref{fig:west_inclined_FM_UF} shows,
we observed a slight decrease in logical error rate
but it was not enough to explain the accuracy discrepancy.

\begin{figure}
	\centering
	\includegraphics[width=0.48\textwidth]{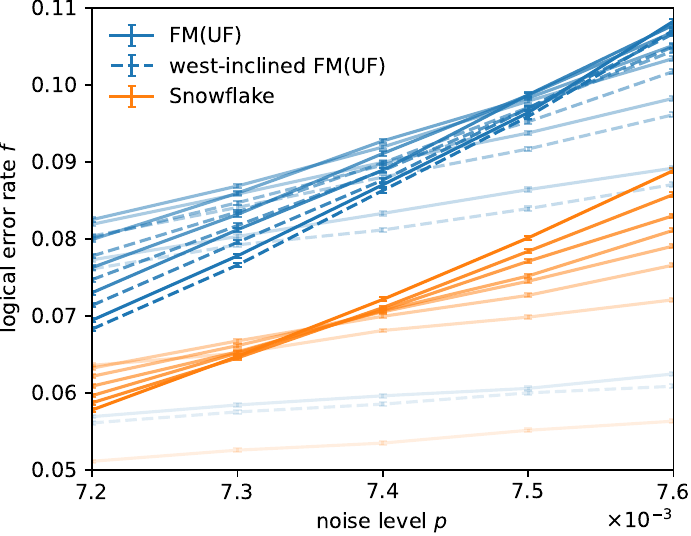}
	\caption{A zoomed-in view of \cref{fig:FM_UF_v_Snowflake}.
	The additional decoder,
	west-inclined FM(UF),
	is the same as FM(UF) but whenever a cluster touches both boundaries,
	it always pushes defects within that cluster to the west boundary,
	mimicking the behaviour of Snowflake.}
	\label{fig:west_inclined_FM_UF}
\end{figure}

\end{document}